\documentclass[letterpaper, 10 pt, conference]{ieeeconf}  % Comment this line out if you need a4paper

\IEEEoverridecommandlockouts                              % This command is only needed if 
\usepackage{amsmath} % assumes amsmath package installed
\usepackage{amsthm}
\usepackage{cases}

\usepackage{amsfonts}
\usepackage{amssymb}  % assumes amsmath package installed
\usepackage{mathtools} %% for \coloneqq
\usepackage{mathrsfs}
\usepackage{balance}
\usepackage{nicefrac}
\DeclareMathAlphabet{\mathpzc}{OT1}{pzc}{m}{it}

\newcommand{\R}{\mathbb{R}}
\newcommand{\E}{\mathbb{E}}

\theoremstyle{definition}

\newtheorem{definition}{Definition}

\newtheorem{lemma}{Lemma}
\newtheorem{theorem}{Theorem}
\newtheorem{remark}{Remark}

\usepackage{cleveref}

\usepackage[ruled,linesnumbered]{algorithm2e}

\newcommand{\norm}[1]{\left\lVert #1\right\rVert}
\newcommand{\set}[1]{\left\{ #1\right\}}

\DeclareMathOperator*{\argmin}{arg\,min}

\newcommand{\doi}[1]{\href{http://dx.doi.org/#1}{\normalsize{\textsc{doi:}}~\nolinkurl{#1}}}
\newcommand{\arxiv}[1]{\href{http://arxiv.org/abs/#1}{\normalsize{\textsc{arxiv:}}~\nolinkurl{#1}}}
\usepackage{color}

\usepackage{comment}
\newcommand{\HRule}{\noindent\rule{\linewidth}{0.1mm}\newline}

\renewcommand{\epsilon}{\varepsilon}
\renewcommand{\phi}{\varphi}

\renewcommand{\implies}{\;\Rightarrow\;}

\newcommand{\gpdomain}{\mathcal{X}}
\newcommand{\gpaugdomain}{\bar{\mathcal{X}}}
\newcommand{\gpvar}{x}
\newcommand{\augu}{y}
\newcommand{\augU}{\mathcal{Y}}

\newcommand{\gpmuV}{\mu_{V}}
\newcommand{\gpsigmaV}{\sigma_{V}}

\newcommand{\gpmuB}{\mu_{B}}
\newcommand{\gpsigmaB}{\sigma_{B}}
\newcommand{\gpGramB}{C_{B}}

\newcommand{\socA}{Q(x)}
\newcommand{\socb}{\mathrm{r}(x)}
\newcommand{\socc}{\mathrm{w}(x)}
\newcommand{\socd}{\mathrm{v}(x)}
\newcommand{\eigval}{\lambda_\dagger}
\newcommand{\eigvec}{e_\dagger}

\newcommand{\fc}[1]{{\color{black} #1}}

\title{\LARGE \bf
Pointwise Feasibility of Gaussian Process-based Safety-Critical Control under Model Uncertainty
}

\author{Fernando Castañeda*, Jason J. Choi*, Bike Zhang, Claire J. Tomlin, and Koushil Sreenath% <-this % stops a space
\thanks{*Indicates equal contribution.}%
\thanks{The authors are with the University of California, Berkeley, CA, 94720, USA, \tt\small \{fcastaneda, jason.choi, bikezhang, tomlin, koushils\}@berkeley.edu}%%
\thanks{This work was partially supported through National Science Foundation Grant CMMI-1931853, and DARPA Assured Autonomy program, grant FA8750-18-C-0101. The work of Fernando Casta\~neda received the support of a fellowship from Fundaci\'on Rafael del Pino, Spain.}
}

\begin{document}

\maketitle
\thispagestyle{empty}
\pagestyle{empty}

%%%%%%%%%%%%%%%%%%%%%%%%%%%%%%%%%%%%%%%%%%%%%%%%%%%%%%%%%%%%%%%%%%%%%%%%%%%%%%%%
\begin{abstract}
Control Barrier Functions (CBFs) and Control Lyapunov Functions (CLFs) are popular tools for enforcing safety and stability of a controlled system, respectively.
They are commonly utilized to build constraints that can be incorporated in a min-norm quadratic program (CBF-CLF-QP) which solves for a safety-critical control input.
However, since these constraints rely on a model of the system, when this model is inaccurate the guarantees of safety and stability can be easily lost.
In this paper, we present a Gaussian Process (GP)-based approach to tackle the problem of model uncertainty in safety-critical controllers that use CBFs and CLFs. The considered model uncertainty is affected by both state and control input.
We derive probabilistic bounds on the effects that such model uncertainty has on the dynamics of the CBF and CLF. We then use these bounds to build safety and stability chance constraints that can be incorporated in a min-norm convex optimization-based controller, called GP-CBF-CLF-SOCP.
As the main theoretical result of the paper, we present necessary and sufficient conditions for pointwise feasibility of the proposed optimization problem.
We believe that these conditions could serve as a starting point towards understanding what are the minimal requirements on the distribution of data collected from the real system in order to guarantee safety. 
Finally, we validate the proposed framework with numerical simulations of an adaptive cruise controller for an automotive system.
\end{abstract}

\section{Introduction}
\label{sec:01introduction}

\subsection{Motivation}

Guaranteeing stability and safety is of major importance for the deployment of autonomous systems in the real world. % Model-based controllers provide a theoretical framework for this purpose,
Theory for model-based controllers for this purpose is well established \cite{ames2019control}.
However, the models that these controllers rely on are usually inaccurate. Real systems can be hard to model, might have parameters that can be challenging to identify, or both. In the presence of such model uncertainty, the theoretical guarantees of stability and safety that model-based controllers provide can be lost.
On the other hand, data-driven control techniques have the potential to solve difficult control tasks and are becoming increasingly popular.
However, guarantees of safety or stability for these approaches typically require rich enough data to fully characterize the dynamics of the system \cite{deepc, pilco, verhaegen2007filtering}. In practice, collecting these data might not be feasible for systems with uncertain dynamics and uncertain input effects. For instance, it might require applying unstable control inputs that would lead to unsafe behavior and possible damage of the system. Providing safety guarantees for such uncertain systems while using only data that is feasible to collect remains an open problem \cite{taylor2020towards}.

This paper takes important steps towards solving this problem by combining a model-based method with a data-driven approach.
\fc{First, we build on our previous work \cite{GPCLFSOCP} to devise an optimization-based controller that can address model uncertainty in safety constraints.
We then conduct a feasibility analysis of this controller, which results in a set of conditions on the uncertainty characterization that are necessary and sufficient
to probabilistically guarantee safety pointwise in time.
Indeed, we believe that this analysis is key to verifying whether the available data is sufficient to achieve our controller's safety premise.
We therefore view this work as the potential foundation for many practical data-driven safety-control frameworks, for instance, to guide the online data collection process of safe learning algorithms while maintaining safety.
}

\subsection{Related Work}
Control Barrier Functions (CBFs, \cite{ames2014cbf}) and Control Lyapunov Functions (CLFs, \cite{artstein}) have been popularly used to maintain safety or achieve stability, respectively, of a controlled system.
Model uncertainty of CBF-based controllers was first explored by the use of adaptive \cite{nguyen2017robust, taylor2020adaptive} and robust \cite{nguyen2021robust, jankovic2018robust} control techniques.
Recent work has demonstrated that CBF-based controllers combined with data-driven methods can be useful for coping with model uncertainty. With neural networks it is often difficult to obtain rigorous guarantees for reliable safety controllers, although several works show practical outcomes \cite{choi2020reinforcement, westenbroek2021combining, taylor2020cbf}.
Gaussian Process (GP) regression is an alternative approach that provides a probabilistic guarantee of its prediction.
CLF or CBF-based controllers that use GP regression to learn the model uncertainty terms have been proposed before \cite{berkenkamp2016lyapunov, fan2019balsa, cheng2020safe}.
However, these works do not consider uncertainty in the input effects, which is significant for many systems \cite{GPCLFSOCP}. In \cite{dhiman2020control}, this problem is tackled by using Matrix-Variate GP regression.
\fc{Finally, \cite{lederer2021impact} analyzes how the choice of data affects the GP prediction uncertainty and the performance of their learning-based stabilizing controller.
However,  by  fixing  a closed-loop policy, the uncertainty is only characterized with respect to the state, but not the control input.

}

\subsection{Contributions}
\fc{First, we extend our previous work \cite{GPCLFSOCP}, where we used GP regression to address model uncertainty in CLFs, to the safety-critical control case by including an uncertainty-aware CBF chance constraint.
This results in the formulation of the so-called \textit{Gaussian Process-based Control Barrier Function and Control Lyapunov Function Second-Order Cone Program (GP-CBF-CLF-SOCP)} controller. The main theoretical result of this paper is a set of necessary, sufficient, and necessary and sufficient conditions for pointwise feasibility---feasibility at each timestep---of the proposed GP-CBF-CLF-SOCP controller.
The interpretation of these conditions as an interplay between achieving the safety goal and reducing the GP prediction error constitutes the main theoretical insight derived in this paper.}

\section{Background}
\label{sec:02background}
In this section, we present the necessary background on CLFs and CBFs.

Consider a nonlinear control-affine system of the form:
\vspace{-5pt}
\begin{equation}\label{eq:system}
    \dot{x} = f(x) + g(x)u,
    \vspace{-5pt}
\end{equation}
where $x \in \mathcal{X} \subset \R^n$ is the system state and $u \in \mathbb{R}^m$ is the control input. 
The vector fields $f$ and $g$ are assumed to be locally Lipschitz continuous, and without loss of generality we assume that the origin is the equilibrium point, $f(0) = 0$,
that we want the system state to converge to.
We will refer to system \eqref{eq:system} as the \emph{true plant}.

%%%%%%%%%%%%%%%%%%%%%%%%%%%%%%%%%%%%%%%%%%%%%%%%%%%%%%%%%%%%
\subsection{Control Lyapunov Functions}
\label{subsec:0201clf}

\begin{definition}
\label{def:clf}
Let $V \colon \mathcal{X} \to \R_{+}$ be a positive definite, continuously differentiable, and radially unbounded function. $V$ is a \emph{Control Lyapunov Function} (CLF) for system \eqref{eq:system} if for each $x \in \mathcal{X} \setminus \set{0}$ the following holds:
\vspace{-3pt}
\begin{equation}\label{eq:clf_def}
    \inf_{u \in \mathbb{R}^m} \underbrace{L_f V(x) + L_g V(x)u}_{= \dot{V}(x,u)} < 0,
\vspace{-3pt}
\end{equation}
where $L_f V(x)\!\coloneqq\!\nabla V(x)\!\cdot\!f(x)$ and $L_g V(x)\!\coloneqq\!\nabla V(x)\!\cdot\!g(x)$ are Lie derivatives of $V$ with respect to $f$ and $g$, respectively.
\end{definition}
If system \eqref{eq:system} admits such a CLF, it is globally asymptotically stabilizable to the origin \cite{artstein}.
CLFs can also impose a stronger notion of stabilizability, which is the exponential convergence to the origin.
If there exists a compact subset $D \subseteq \mathcal{X}$ that includes the origin such that for all $ x \in D$ and for some constant $\lambda>0$, it holds that
\vspace{-3pt}
\begin{equation}
\label{eq:expclf_local_def}
    \inf_{u \in \R^m} L_f V(x) + L_g V(x)u + \lambda V(x) \leq 0,
    \vspace{-3pt}
\end{equation}
and $\exists \ c_{exp}>0$ such that $ \Omega_{c_{exp}}\!\coloneqq \!\{x\!\in\!D \subseteq\mathcal{X}: V(x)\!\leq\!c_{exp}\}$ is a sublevel set of $V(x)$, then the origin is locally exponentially stabilizable from $\Omega_{c_{exp}}$. In this case, we say that $V$ is a \textit{locally exponentially stabilizing} CLF. Condition
\eqref{eq:expclf_local_def} can be used as a constraint in a min-norm quadratic program (QP) \cite{ames2013clfqp}:

\small
\HRule
\noindent \textbf{CLF-QP}:
\begin{subequations}
\label{eq:clf-qp-all}
\begin{align}
u^{*}(x) & = & & \underset{u\in \R^m}{\argmin}  \quad \norm{u}_2^2 \label{eq:clf-qp}\\
& \text{s.t.} & & L_f V(x) + L_g V(x)u + \lambda V(x) \leq 0. \label{eq:eclf}
\end{align}
\end{subequations}
\hrule
\normalsize
\vspace{2mm}
The above QP defines a min-norm feedback control law $u^* \colon \mathcal{X} \to \R^m$ that renders the origin exponentially stable.

%%%%%%%%%%%%%%%%%%%%%%%%%%%%%%%%%%%%%%%%%%%%%%%%%%%%%%%%%%%%
\subsection{Control Barrier Functions}
\begin{definition}
\label{def:cbf} Consider a continuously differentiable function $B: \mathcal{X} \subset \R^n \to \R$ and a set $\mathcal{C}$ defined as the superlevel set of $B$, $\mathcal{C} = \{ x \in \mathcal{X}: B(x) \geq 0 \}$. $B$ is a \emph{Control Barrier Function} (CBF) for system \eqref{eq:system} if there exists an extended class $\mathcal{K}_\infty$ function $\gamma$ such that for all $x \in \mathcal{X}$,
\vspace{-3pt}
\begin{equation}
    \label{eq:cbf_def}
    \sup_{u \in \R^m} \underbrace{L_f B(x) + L_g B(x) u}_{= \dot{B}(x,u)} + \gamma (B(x)) \geq 0.
    \vspace{-3pt}
\end{equation}
\end{definition}
If $B$ is a CBF for system \eqref{eq:system} and {\small $\nabla B(x) \neq 0$} for all $x \in \partial \mathcal{C}$, any Lipschitz continuous control input $u$ satisfying \eqref{eq:cbf_def} renders the set $\mathcal{C}$ forward invariant \cite[Thm. 2]{ames2019control}.
In \cite{ames2014cbf}, a QP formulation of a safety-critical controller is proposed by incorporating both conditions \eqref{eq:expclf_local_def} and \eqref{eq:cbf_def} as constraints:

\small
\HRule
\noindent \textbf{CBF-CLF-QP}:
\begin{subequations}
\label{eq:cbf-clf-qp-all}
\begin{align}
u^{*}(x) & = & & \underset{(u, d)\in \R^{m+1}}{\argmin}  \quad \norm{u}_2^2 + p d^2 \label{eq:cbf-clf-qp}\\
& \text{s.t.} & & L_f V(x) + L_g V(x)u + \lambda V(x) \leq d, \label{eq:cbfclf-constraint1} \\
&  & & L_f B(x) + L_g B(x)u + \gamma (B(x)) \geq 0, \label{eq:cbfclf-constraint2} \vspace{-.5em}
\end{align}
\end{subequations}
\hrule
\normalsize
\vspace{2mm}
\noindent where $d$ is a slack variable used to relax the CLF constraint in order to give preference to safety over stability in case of conflict.
In this paper, we will refer to \eqref{eq:cbfclf-constraint1} as the \textit{(relaxed) CLF constraint} and to \eqref{eq:cbfclf-constraint2} as the \textit{CBF constraint}.
Note that the Lie derivatives of $V$ and $B$, which appear in the these constraints, require explicit knowledge of the dynamics of the plant.

%%%%%%%%%%%%%%%%%%%%%%%%%%%%%%%%%%%%%%%%%%%%%%%%%%%%%%%%%%%%
\subsection{Adverse Effects of Model Uncertainty}
\label{subsec:02adverse-effects}

We now provide some necessary settings and assumptions for our problem formulation.
Let's assume that we have a \textit{nominal model}: \vspace{-.3em}
\begin{equation}
    \label{eq:nominal-model}
    \dot{x} = \tilde{f}(x) + \tilde{g}(x)u, \vspace{-.3em}
\end{equation}
where $\tilde{f}$ and $\tilde{g}$ are Lipschitz continuous vector fields and $\tilde{f}(0)=0$. In general, the nominal model vector fields ($\tilde{f}$, $\tilde{g}$) do not perfectly match the true plant vector fields ($f$, $g$) due to model uncertainty.

We start by designing a locally exponentially stabilizing CLF $V$ and a CBF $B$, based on the nominal model \eqref{eq:nominal-model}. These functions are assumed to be a locally exponentially stabilizing CLF and a CBF, respectively, also for the true plant \eqref{eq:system}. This is a structural assumption, and it is met for feedback linearizable systems if the nominal model has the same degree of actuation as the true plant.
A more detailed explanation can be found in \cite{taylor2020towards}. Finally, we also assume that we can measure the state $x$.

The main objective of the paper is to learn the CLF constraint \eqref{eq:cbfclf-constraint1} and the CBF constraint \eqref{eq:cbfclf-constraint2} for the true plant, in order to be able to synthesize safe and stabilizing controllers.
Note that the actual derivatives of $V$ and $B$ depend on the true plant dynamics:
\vspace{-1em}

\small
\begin{align*}
    \dot{V}(x,u)\!=\!L_f V(x)\!+\!L_g V(x)u, \quad \dot{B}(x,u)\!=\!L_f B(x)\!+\!L_g B(x)u.
\end{align*}
\normalsize

\noindent However, the nominal model-based estimates of their values are given by:
\vspace{-1em}

\small
\begin{align*}
    \tilde{\dot{V}}(x, u)\!=\!L_{\tilde{f}}V(x)\!+\!L_{\tilde{g}}V(x) u, \quad \tilde{\dot{B}}(x, u)\!=\!L_{\tilde{f}}B(x)\!+\!L_{\tilde{g}}B(x) u,
\end{align*}
\normalsize
\noindent and can differ from the true values. We define $\Delta_V,\ \Delta_B: \mathcal{X} \times \R^m \rightarrow \R$ as the errors of the nominal model-based estimates:
\vspace{-.5em}
\begin{align}
    \Delta_V(x, u) := \dot{V}(x, u) - \tilde{\dot{V}}(x, u), \label{eq:mismatch_clf} \\
    \Delta_B(x, u) := \dot{B}(x, u) - \tilde{\dot{B}}(x, u). \label{eq:mismatch_cbf} \vspace{-1em}
\end{align}
Then, the CLF \eqref{eq:cbfclf-constraint1} and CBF \eqref{eq:cbfclf-constraint2} constraints become \vspace{-.4em}
\begin{align}
    \label{eq:clf-constraint-uncertainty}
    & L_{\tilde{f}}V(x) + L_{\tilde{g}}V(x) u + \Delta_V(x, u) + \lambda V(x)\le d, \\
    \label{eq:cbf-constraint-uncertainty}
    & L_{\tilde{f}}B(x) + L_{\tilde{g}}B(x) u + \Delta_B(x, u) + \gamma(B(x)) \ge 0. 
    \vspace{-1em}
\end{align}
Therefore, by learning the uncertainty terms $\Delta_V(x, u)$ and $\Delta_B(x, u)$ correctly, we can have a good representation of the true CLF and CBF constraints, which can be used to construct a safe stabilizing controller for the true plant. Note that the uncertainty terms $\Delta_V(x, u)$ and $\Delta_B(x, u)$ are scalar values, leading to a lower dimensional learning problem than if we learned the dynamics $f$, $g$. We can approximately measure $\Delta_V(x, u)$ and $\Delta_B(x, u)$ by collecting trajectories from the true plant. These measurements are given by \vspace{-.5em}
\begin{align}
    & z_j^{V} = {\left(V(x(t+\small{\Delta}t))-V(x(t))\right)} / {\small{\Delta}t} - \tilde{\dot{V}}(x_j, u_j), \\
    & z_j^{B} = {\left(B(x(t+\small{\Delta}t))-B(x(t))\right)} / {\small{\Delta}t} - \tilde{\dot{B}}(x_j, u_j), \vspace{-.5em}
\end{align}
where $x_j = {(x(t+\small{\Delta}t)\!+\!x(t))}/2$ is the mean of the state between $[t, t\!+\!\Delta t)$, and $u_j$ is the control input during the same interval; $\Delta t$ is the sampling interval; $z_j^{V}$ and $z_j^{B}$ are approximate measurements of $\Delta_V(x_j, u_j)$ and $\Delta_B(x_j, u_j)$, for $j=1,\ldots,N$. Based on these data, we can define the regression problems for $\Delta_V$ and $\Delta_B$ as supervised learning problems. We will specifically use GP regression.

We close this section by revealing the control-affine structures of $\Delta_V$ and $\Delta_B$. In \eqref{eq:mismatch_clf} and \eqref{eq:mismatch_cbf}, if we express $\dot{V},\ \dot{B}$ and $\tilde{\dot{V}},\ \tilde{\dot{B}}$ with their respective Lie derivatives, we get
\vspace{-1em}

{\setlength{\belowdisplayskip}{0pt}
\small
\begin{align}
    \Delta_V(x, u) = \Phi_V \begin{bmatrix}1 \\ u\end{bmatrix} := (L_f V \!-\!L_{\tilde{f}}V)(x) + (L_g V\!-\!L_{\tilde{g}}V)(x)u, \label{eq:mismatchaffine_clf}
 \\
    \Delta_B(x, u) = \Phi_B \begin{bmatrix}1 \\ u\end{bmatrix} := (L_f B \!-\!L_{\tilde{f}}B)(x) + (L_g B\!-\!L_{\tilde{g}}B)(x)u, \label{eq:mismatchaffine_cbf} \vspace{-.5em}
\end{align}
}
% \vspace{-.5em}
\normalsize

\noindent where $\Phi_V, \Phi_B \in \R^{1\times(m+1)}$. To simplify notations, we will introduce $\augu:=[1 \;\ u^T]^T$. 
Then, both $\Delta_V$ and $\Delta_B$ are obtained by taking a dot product between a $(m+1)$ dimensional vector field ($\Phi_V$ or $\Phi_B$) and $\augu$. This property will be used to construct a suitable kernel structure for the GP regression in the next section. 

% \newpage

\section{Gaussian Process Regression}
\label{sec:03gp}

This section gives an overview of Gaussian Processes and how they are used for regression in this paper. We utilize a compound kernel, which was introduced in our previous work \cite{GPCLFSOCP}, to exploit the affine nature of the problem.

\subsection{Gaussian Processes}
A Gaussian Process (GP) is a type of random process such that any finite collection of its samples always has a joint Gaussian distribution. It is characterized by two functions: a mean function $q: \gpdomain\rightarrow\R$ and a covariance function $k: \gpdomain\times\gpdomain\rightarrow\R$, where $\gpdomain$ is the input domain of the process. We express the process as: \vspace{-.3em}
\begin{equation}
\label{eq:gp_def}
    h(\gpvar) \sim \mathcal{GP}(q(\gpvar), k(\gpvar,\gpvar')), \vspace{-.3em}
\end{equation}
where $h(\cdot)$ is a sample drawn from it, and for any $\gpvar, \gpvar' \in \gpdomain$
\begin{equation*}
    \E[h(\gpvar)] = q(\gpvar), \; \E[(h(\gpvar)-q(\gpvar))(h(\gpvar')-q(\gpvar'))] = k(\gpvar, \gpvar').
\end{equation*}
Since the joint distribution of $h(\gpvar)$ and $h(\gpvar')$ is Gaussian, the functions $q$ and $k$ fully characterize the process. Note that the covariance function $k$ determines the covariance matrix of the joint Gaussian distribution of finite samples. Therefore, only those functions that lead to positive semidefinite covariance matrices can be used for $k$. If $k$ satisfies this requirement, it is said to be a positive definite kernel \cite{wendland2004scattered}.

Given a set of finite measurements of the form $\{(\gpvar_j,h(\gpvar_j)+\epsilon_j)\}_{j=1}^{N}$, where $\epsilon_j \sim \mathcal{N}(0, \sigma_{n}^2)$ is white measurement noise, and a query point $\gpvar_{*}$, a posterior distribution for $h(\gpvar_{*})$ can be derived from the Gaussian distribution of $[h(\gpvar_1), \cdots, h(\gpvar_N), h(\gpvar_{*})]^T$ conditioned on the measurements. This can be used as a prediction of $h(\gpvar_{*})$ at a query point $\gpvar_*$, with mean and variance
\vspace{-.5em}
\begin{equation}
\label{eq:gpposteriormu}
        \mu_{*} = \mathbf{z}^T (K + \sigma_n^2 I )^{-1} K_{*}^{T}, \vspace{-.5em}
\end{equation}
\begin{equation}
\label{eq:gpposteriorsigma}
    \sigma_{*}^{2} = k\left(\gpvar_{*}, \gpvar_{*}\right)-K_{*}  (K + \sigma_n^2 I )^{-1} K_{*}^{T}, \vspace{-.3em}
\end{equation}
where $K\in\R^{N\times N}$ is the Gram matrix whose $(i, j)^{th}$ element is $k(\gpvar_i, \gpvar_j)$, $K_{*}=[k(\gpvar_*, \gpvar_1),\ \cdots \ ,k(\gpvar_*, \gpvar_N)]\in\R^N$, and $\mathbf{z}\in \R^N$ is the vector containing the output measurements $z_j:=h(\gpvar_j)+\epsilon_j$. Here, we used $q(\gpvar)\!\equiv\!0$ as the mean function in \eqref{eq:gp_def}. The fact that this closed form expression of the prediction is available makes GPs attractive for many regression problems \cite{williams2006gaussian}. 

\subsection{GP Regression with the Affine Dot Product Kernel}

Our goal is to use GP regression for $\Delta_{V}(x, u)$ and $\Delta_{B}(x, u)$. Note that the input domain is now $\mathcal{X} \times \R^{m+1}$ instead of $\gpdomain$, where $\R^{m+1}$ is the space of $\augu=[1 \;\; u^T]^T$.
We define $\gpaugdomain := \mathcal{X} \times \R^{m+1}$. Even though any positive definite kernel in $\gpaugdomain$ can be used for the regression, we would like to exploit the fact that both $\Delta_{V}$ and $\Delta_{B}$ are affine in $\augu$.
The following compound kernel structure is used for this goal.
\begin{definition} \label{def:adpkernel}
\emph{Affine Dot Product Compound Kernel: }
Define $k_{c}:\gpaugdomain \times \gpaugdomain \rightarrow \R$ given by
\vspace{-1em}

\small
\begin{equation}
    k_{c}\left(\left[\begin{array}{c} x \\ y \\\end{array}\right], \left[\begin{array}{c} x' \\ y' \\\end{array}\right]\right) := y^T Diag([k_1(x, x'), \cdots, k_{m+1}(x, x')]) y'
\label{eq:adpkernel}
\end{equation}
% }
\normalsize

\noindent as the \emph{Affine Dot Product} (ADP) compound kernel of $(m\!+\!1)$ individual kernels $k_1,\ldots ,$ $k_{m+1}:\mathcal{X}\times \mathcal{X} \rightarrow \R$ \cite{GPCLFSOCP}.
\end{definition}

\fc{The above compound kernel captures the appropriate structure of the target functions, which results in a much better regression fit compared to using arbitrary kernels. 
Moreover, it results in expressions for the mean and variance of the GP regression that are linear and quadratic in the control input, respectively.}
This is crucial for the construction of the convex min-norm controller that will be introduced in the next section. Let $X\in\R^{n\times N}$ and $Y\in \R^{(m+1) \times N}$ be matrices whose column vectors are the inputs $x_j$ and $\augu_j$ of the collected data, respectively. Then, when using the ADP compound kernel, Equations \eqref{eq:gpposteriormu} and \eqref{eq:gpposteriorsigma} give the following expressions for the mean and variance of the GP prediction at a query point $(x_{*}, \augu_{*})$:
\vspace{-3pt}
\begin{equation}
\label{eq:mu_adp}
    \mu_{*} = \underbrace{\mathbf{z}^T (K_c + \sigma_n^2 I )^{-1} K_{*Y}^{T}}_{=:b_{*}^T} y_{*},
\end{equation}
\vspace{-5pt}
\begin{equation}
\label{eq:sigma_adp}
\small{
    \sigma_{*}^{2} \!= \!y_{*}^{T}\!\underbrace{\left(\!Diag \text{\footnotesize $\left(\!\begin{bmatrix}k_1(x_{*}, x_{*}) \\ \vdots \\ k_{m+1}(x_{*}, x_{*}) \end{bmatrix}\!\right)$} \!-\!K_{*Y}\!(K_c + \sigma_n^2 I )^{-1} K_{*Y}^{T}\! \right)}_{=:C_{*}}\!y_{*}.
}
\end{equation}
Here, $K_c\in\R^{N\times N}$ is the Gram matrix of $k_c$ for the training data inputs ($X,Y$), and $K_{*Y}\in\R^{(m+1)\times N}$ is given by
\begin{equation*}
\small{
    K_{*Y}\!=\!\begin{bmatrix} K_{1*} \\ K_{2*} \\ \vdots \\K_{(m+1)*}
    \end{bmatrix}\!\circ\!Y,\;\;K_{i*}\!=\![k_i(x_{*}, x_1),\cdots, k_i(x_{*}, x_N)].
}
\end{equation*}

\subsection{Derivation of Confidence Bounds for $\Delta_V$ and $\Delta_B$}

In regression problems, the set of candidate functions, which is a design choice, determines the functions' expressivity at the cost of the complexity of the problem. In the machine learning literature, a Reproducing Kernel Hilbert Space (RKHS, \cite{wendland2004scattered}) is a common choice for this set since it encompasses a space of well-behaved and expressive-enough functions. An RKHS, denoted as $\mathcal{H}_k(\gpaugdomain)$, is characterized by a specific positive definite kernel $k$. The kernel $k$ evaluates whether a member of $\mathcal{H}_k(\gpaugdomain)$ satisfies a specific property, namely a ``reproducing" property: the inner product between any member $h\in\mathcal{H}_k(\mathcal{\gpaugdomain})$ and the kernel $k(\cdot, \bar{x})$ should reproduce $h$, i.e., $\langle h(\cdot), k(\cdot, \bar{x}) \rangle_{k} = h(\bar{x}), ~\forall \bar{x}\in \gpaugdomain$. Furthermore, the RKHS norm $\norm{h}_{k}:=\sqrt{\langle h, h \rangle}_{k}$
is a measure of how ``well-behaved'' \footnote{$\norm{h(\bar{x})-h(\bar{x}')}_{2}\le\norm{h}_{k}\norm{k(\bar{x},\cdot)-k(\bar{x}',\cdot)}_{k} \; \forall \bar{x}, \bar{x}' \in \mathcal{\gpaugdomain}$} the function $h\in\mathcal{H}_k(\mathcal{\gpaugdomain})$ is.

It turns out that if a target function $h$ is a member of $\mathcal{H}_k(\gpaugdomain)$ with bounded RKHS norm, then a confidence bound for the true value of $h(\bar{x})$ can be specified by the mean and variance of the GP prediction \cite{gpucb}. We can apply this confidence bound analysis to GP regression problems that use the ADP compound kernel $k_c$ \eqref{eq:adpkernel} as the reproducing kernel. Let $\Delta$ be either $\Delta_{V}$ or $\Delta_{B}$, and $\Phi$ be either $\Phi_V$ or $\Phi_B$ from \eqref{eq:mismatchaffine_clf} and \eqref{eq:mismatchaffine_cbf}. The following theorem provides a probabilistic bound on the estimation error of $\Delta$ from the collected data.

\begin{theorem}
\label{theorem:DeltaUCB}
\cite[Thm.~2]{GPCLFSOCP} Consider $m\!+\!1$ bounded kernels $k_{i}$, for $i\!=\!1,\ldots,(m+1)$. Assume that each $i$-th element of $\Phi$ is a member of $\mathcal{H}_{k_{i}}$ with bounded RKHS norm. We also assume that we have access to measurements $z_i = \Delta(x_i, u_i) + \epsilon_i$, and that each noise term $\epsilon_i$ is zero-mean and uniformly bounded by $\sigma_{n}$. Let $\beta \coloneqq \left(2\eta^2 + 300 \kappa_{N+1} \ln^3((N+1)/\delta)\right)^{0.5}$, with $N$ the number of data points, $\eta$ the bound of $\norm{\Delta}_{k_c}$, and $\kappa_{N+1}$ the maximum information gain after getting $N+1$ data points. Let $\mu_{*}$ and $\sigma^2_{*}$ be the mean \eqref{eq:mu_adp} and variance \eqref{eq:sigma_adp} of the GP regression for $\Delta$, using the ADP compound kernel $k_c$ of $k_1,\cdots,k_{m+1}$, at a query point $(x_{*},\ \augu_{*}=[1\ u_{*}^T]^T)$, where $\augu_{*}$ is an element of a bounded set $\augU\subset\R^{m+1}$. Then, with a probability of at least $1-\delta$ the following holds for all $N\geq 1$, $x_{*} \in \mathcal{X}$, $[1\ u_{*}^T]^T \in \augU$:
\vspace{-3pt}
\begin{equation}
\label{eq:sigmaUCB}
    | \mu_{*} - \Delta(x_{*},u_{*}) | \leq \beta \sigma_{*}.
\vspace{-3pt}
\end{equation}
\end{theorem}

\noindent For technical details of Theorem \ref{theorem:DeltaUCB}, please refer to \cite[Thm.~6]{gpucb}, and our previous paper \cite{GPCLFSOCP}.

\section{Proposed Safety Controller}
\label{sec:04controllers}
In this section, we make use of the probability bounds given by Theorem \ref{theorem:DeltaUCB} to build model uncertainty-aware CBF and CLF chance constraints that can be incorporated in a min-norm optimization problem.

From Theorem \ref{theorem:DeltaUCB}, each of the following inequalities hold with a probability of at least $1-\delta$ for $\forall x \!\in\!\mathcal{X}$, $[1\ u^T]^T \in \augU$:
{\setlength{\belowdisplayskip}{2pt}
\setlength{\abovedisplayskip}{2pt}
\begin{align}
\vspace{-8pt}
\label{eq:V_UCB}
    \dot{V}(x,u) & \leq \tilde{\dot{V}}(x,u) + \gpmuV(x,u) + \beta \gpsigmaV(x,u),\\
\label{eq:B_UCB}
    \dot{B}(x,u) & \geq \tilde{\dot{B}}(x,u) + \gpmuB(x,u) - \beta \gpsigmaB(x,u),
    \vspace{-2pt}
\end{align}
% \vspace{-5pt}
where $\gpmuV(x,u),\ \gpsigmaV(x,u)$ are the mean and standard deviation of the GP prediction of $\Delta_V$ at the query point $(x,u)$, and $\gpmuB(x,u),\ \gpsigmaB(x,u)$ are the mean and standard deviation of the GP prediction of $\Delta_B$ at the query point $(x,u)$, given by Equations \eqref{eq:mu_adp} and \eqref{eq:sigma_adp}. 
}

In our previous work \cite{GPCLFSOCP}, we construct a Second-Order Cone Program (SOCP) by using \eqref{eq:V_UCB} to build an exponential CLF chance constraint. This SOCP defines a pointwise min-norm stabilizing feedback control law $u^* \colon \R^n \to \R^m$:

\small
% \HRule
\vspace{2mm}
\hrule
\vspace{2mm}
\noindent \textbf{GP-CLF-SOCP}:
\begin{align}
u^{*}(x) & = & & \underset{u\in \R^m}{\argmin} \norm{u}_2^2
\label{eq:gp-clf-socp} \\
& \text{s.t.} & & \tilde{\dot{V}}(x,u)\!+\!\gpmuV(x,u)\!+\!\beta \gpsigmaV(x,u)\!+\!\lambda V(x) \leq 0. \nonumber
\end{align}
% \HRule
\hrule
\normalsize
\vspace{2mm}
In this paper, we additionally incorporate a CBF chance constraint using the probabilistic bound on the CBF derivative given by \eqref{eq:B_UCB}. For that, we substitute this bound into the CBF constraint of \eqref{eq:cbfclf-constraint2}, resulting in a chance constraint that, if feasible, guarantees that the true CBF constraint of \eqref{eq:cbfclf-constraint2} is satisfied with a probability of $1-\delta$. This chance constraint is incorporated into the following optimization problem, forming a safety-critical stabilizing feedback controller that takes into account the learned model uncertainty:

\small
% \HRule
\vspace{2mm}
\hrule
\vspace{2mm}
\noindent \textbf{GP-CBF-CLF-SOCP}:
\begin{subequations}
\label{eq:gp-cbf-clf-socp}
\begin{align}
u^{*}(x)& = & &  \hspace{-5pt}\underset{(u, d)\in \R^{m+1}}
{\argmin}\norm{u}_2^2 + p~d^2
 \\
& \text{s.t.} & & \hspace{-5pt}\tilde{\dot{V}}(x,u)\!+\!\gpmuV(x,u)\!+\!\beta \gpsigmaV(x,u)\!+\!\lambda V(x) \leq d, \label{eq:socp-clf-constraint} \\
&&& \hspace{-5pt}\tilde{\dot{B}}(x,u)\!+\!\gpmuB(x,u)\!-\!\beta \gpsigmaB(x,u)\!+\!\gamma (B(x)) \geq 0. \label{eq:socp-cbf-constraint}
\end{align}
\end{subequations}
% \HRule
\hrule
\normalsize
\vspace{2mm}
The above optimization problem is an SOCP if the ADP compound kernel of Definition \ref{def:adpkernel} is used for both $\Delta_V$ and $\Delta_B$ \cite[Thm.~3]{GPCLFSOCP}, thanks to the affine and quadratic structures of the means and variances, respectively. Also, it can be expressed in the standard form of SOCPs by transforming the quadratic objective into a Second-Order Cone (SOC) constraint and a linear objective, since the two chance constraints are SOC constraints.

This program is solved in real time to obtain a safe stabilizing control input at the current state. Note that the CLF constraint is relaxed in order to give preference to safety over stability in case of conflict. Also, the GP-CBF-CLF-SOCP does not need knowledge about the true plant dynamics, since it uses the estimates of $\Delta_V$ and $\Delta_B$ obtained from the GP regression.
\section{Analysis of Pointwise Feasibility}
% of the GP-CBF-CLF-SOCP}
\label{sec:05feasibility}

For the feasibility analysis of the GP-CBF-CLF-SOCP, we only have to consider the single hard constraint of the problem, which is the CBF chance constraint \eqref{eq:socp-cbf-constraint},
\begin{equation*}
    L_{\tilde{f}}B(x) + L_{\tilde{g}}B(x)u + \gpmuB(x,u) - \beta \gpsigmaB(x,u) + \gamma (B(x)) \geq 0. \nonumber
\end{equation*}
% \fc{
\begin{remark}
\fc{By studying the conditions under which the GP-CBF-CLF-SOCP is feasible, we are actually trying to answer the question of how accurately we need to know our system in order to guarantee safety with a certain probability. Note that unlike QPs, SOCPs with even a single hard constraint can be infeasible.
In our case, it becomes infeasible when the prediction uncertainty is so high that it eclipses the discovery of a control input that can guarantee the system's safety.}
\end{remark}
% }
From \eqref{eq:mu_adp} and \eqref{eq:sigma_adp} we can obtain the mean $\mu_B$ and variance $\sigma^2_B$ of the GP prediction of $\Delta_B$: \begin{equation}
\label{eq:gp_mu_var_feas}
    \hspace{-3pt}\gpmuB (x,u)\!=\!b_{B}(x)^T \! \begin{bmatrix}1 \\ u\end{bmatrix}\!,\; \gpsigmaB^2(x,u)\!=\![1 \;\; u^T] \gpGramB(x)\!\begin{bmatrix}1 \\ u\end{bmatrix}.
\end{equation}
Here we use $b_{B}$ and $C_B$ instead of $b_{*}$ and $C_*$ to make it clear that we refer to the prediction of the CBF uncertainty $\Delta_B$.
Note that when $k$ is a valid kernel and $\sigma_n > 0$, the Gram matrix of the variance, $\gpGramB(\cdot)$ is always positive definite. We can therefore write
\vspace{-5pt}
\begin{equation}
\label{eq:cone_feas}
    \gpsigmaB(x, u) = \norm{G(x) \begin{bmatrix}1 \\ u\end{bmatrix}},
\end{equation}
where $G(\cdot)\in\R^{(m+1)\times(m+1)}$ is the matrix square root of $\gpGramB(\cdot)$.
Rewriting \fc{the CBF chance constraint} \eqref{eq:socp-cbf-constraint} in terms of \eqref{eq:gp_mu_var_feas} and \eqref{eq:cone_feas} gives
\begin{equation}
\label{eq:intermediate}
\footnotesize
    \beta \norm{G \begin{bmatrix}1 \\ u\end{bmatrix}} \leq L_{\tilde{f}}B(x) + b_{B}(x)_{1} + \gamma (B(x)) + L_{\tilde{g}}B(x)u + b_{B}(x)_{2:(m+1)}^T u.
\end{equation}
\normalsize
The numerical subscripts are used to indicate elements of vectors or columns of matrices.
\fc{We can now express \eqref{eq:socp-cbf-constraint} in the standard form of SOC constraints:}
\vspace{-2pt}
\begin{equation}
\label{eq:constraint_std}
    \norm{\socA u + \socb} \le \socc u + \socd,
\end{equation}
with $Q\coloneqq \beta G_{2:(m+1)} \in \R^{(m+1)\times m}$, $r \coloneqq \beta G_{1}\in \R^{(m+1)\times1}$,
\vspace{-15pt}
\begin{align}
    & \hspace{-5pt}\socc = \widehat{L_{g}B}(x) := L_{\tilde{g}}B(x) + b_{B}(x)_{2:(m+1)}^T \in \R^{1\times m}, \\
    & \hspace{-5pt}\socd = \widehat{L_{f}B}(x)+ \gamma (B(x)) \in \R,
\end{align}
where $\widehat{L_{f}B}(x):=L_{\tilde{f}}B(x) + b_{B}(x)_{1}$.
$\widehat{L_{g}B}(x)$ is the mean of the prediction of $L_{g}B(x)$, as it incorporates the model-based term $L_{\tilde{g}}B(x)$ and the mean of the GP prediction of the actuation mismatch $b_{B}(x)_{2:(m+1)}^T$. Similarly, $\widehat{L_{f}B}(x)$ is the mean of the prediction of $L_{f}B(x)$. Note that
\begin{equation}
\vspace{-5pt}
\label{eq:eq-mean-prediction}
    [\socd\;\socc] \begin{bmatrix}1 \\ u\end{bmatrix} = \widehat{L_{f}B}(x) + \widehat{L_{g}B}(x) u + \gamma (B(x))
\end{equation}
is the mean prediction of the left-hand side of the true plant CBF constraint \eqref{eq:cbfclf-constraint2}.

We now provide an equivalent formulation of constraint \eqref{eq:constraint_std}, which will be useful for the feasibility analysis.
\begin{lemma}
\label{lemma:fundamental}The SOC constraint \eqref{eq:constraint_std} is feasible at $x \in \R^n$ if and only if there exists a $u\in\R^m$ such that both of the following two conditions hold: \vspace{-.5em}
\begin{subequations}\label{eq:feas_lemma1}
\begin{numcases}{}
    [1\ u^T] H(x) \begin{bmatrix} 1 \\u \end{bmatrix} \leq 0,    \label{eq:feas_lemma1_1}\\
    \socc u + \socd \geq 0,     \label{eq:feas_lemma1_2} \vspace{-.5em}
\end{numcases}
\end{subequations}
where the symmetric matrix $H(x)$ takes the form \begin{equation}
\label{eq:H_def}
\small
    H(x)\!=\!\begin{bmatrix} \socb^T\socb\!-\!\socd^2 \!&\!\socb^T\socA\!-\!\socd\socc \\
    \socA^T\socb\!-\!\socc^T\socd\!&\! \socA^T\socA\!-\!\socc^T\socc \end{bmatrix}.
\end{equation}
\end{lemma}
\vspace{-1em}

\begin{proof}
Inequality \eqref{eq:feas_lemma1_1} is obtained by simply taking squares on both sides of \eqref{eq:constraint_std}, and \eqref{eq:feas_lemma1_2} checks that the value of the right-hand side of \eqref{eq:constraint_std} is non-negative, since the left-hand side is always non-negative.
\end{proof}

\subsection{Necessary Condition}

With Lemma \ref{lemma:fundamental}, we can now formulate a 
necessary condition for pointwise feasibility of the GP-CBF-CLF-SOCP.

\begin{lemma}[Necessary condition for pointwise feasibility]
\label{lemma:necessary}
If the GP-CBF-CLF-SOCP of \eqref{eq:gp-cbf-clf-socp} is feasible at a point $x \in \R^n$, then the following condition must be satisfied:
\begin{equation}
\label{eq:necessary}
     [\socd\;\;\socc]\ \gpGramB(x)^{-1}\begin{bmatrix}\socd \\ \socc^T\end{bmatrix} \geq \beta^2,
\end{equation}
or, equivalently, $H(x)$ cannot be positive definite.
\end{lemma}

\begin{proof}
\fc{See Appendix.}
\end{proof}

The left-hand side of \eqref{eq:necessary} encodes a trade-off between the vector $[\socd\ \socc]$---information about the mean prediction of the CBF constraint as \eqref{eq:eq-mean-prediction} indicates---and the uncertainty matrix $\gpGramB(x)$.

\fc{To provide some insight, note that the left-hand side of \eqref{eq:necessary} contains $f$-related components ($\socd=\widehat{L_{f}B}(x)+ \gamma (B(x))$ and the upper left block of $\gpGramB(x)$) and $g$-related components ($\socc = \widehat{L_{g}B}(x)$ and the lower right block of $\gpGramB(x)$).
The $f$-related components encode a trade-off between the mean predicted value of the CBF constraint without actuation ($\socd$) and the unactuated prediction uncertainty (the upper left block of $\gpGramB(x)$).
On the other hand, $\socc$ reflects how $u$ can influence the change of $B(x)$; speaking informally, the dynamics of $B(x)$ would be ``controllable" when $\socc$ is a non-zero vector. In this case, the $g$-related components of \eqref{eq:necessary} reveal that the controllability of $B(x)$ ($\socc$) ``defeating'' the growth of actuation uncertainty with respect to $u$ (the lower right block of $\gpGramB(x)$) helps secure the necessary condition for feasibility. As a matter of fact, this actuation-dependent portion of the trade-off can be used by itself to verify a sufficient condition for feasibility, which will be explained next.}

\subsection{Sufficient Condition}

We now present a sufficient condition for feasibility of the GP-CBF-CLF-SOCP. \fc{As introduced before, this condition is based on the existence of an input direction along which the controllability of $B(x)$ ``defeats'' the growth of the prediction uncertainty with respect to $u$, as will be seen in the proof of Lemma \ref{lemma:sufficient}.}

\begin{lemma}[Sufficient condition for pointwise feasibility]
\label{lemma:sufficient}
For a point $x \in \R^n$, let $\eigval$ be the minimum eigenvalue of the symmetric matrix $\socA^T \socA - \socc^T \socc$.
If $\eigval<0$, the GP-CBF-CLF-SOCP \eqref{eq:gp-cbf-clf-socp} is feasible at $x$.
\end{lemma}

\begin{proof}
For a specific $x \in \R^n$, define \vspace{-.5em}
\begin{equation}
\label{eq:defineF}
    F \coloneqq \socA^T \socA - \socc^T \socc.
\vspace{-.5em}
\end{equation}
Let $\eigvec$ be the unit eigenvector of $F$ associated with the eigenvalue $\eigval$. Then, $\eigval\!<\!0\!\implies\!\eigvec^T F \eigvec\!<\!0$. Since $\socA^T \socA$ is positive definite, this indicates that $\socc \eigvec \neq 0$. Let's take a control input in the direction of this eigenvector:\vspace{-.5em}
\begin{equation}
\label{eq:feasibledirection}
    u = \alpha \cdot \text{sgn}(\socc \eigvec) \cdot \eigvec, \quad \alpha > 0.
    \vspace{-.5em}
\end{equation}
Using this control input, the left-hand side of Equation \eqref{eq:feas_lemma1_1} becomes $(\socb^T\socb - \socd^2) + 2 \alpha\ \! \text{sgn}(\socc \eigvec)(\socb^T\socA - \socd\socc) \eigvec + \alpha^2 \eigvec^T F \eigvec  $, which can be made negative by choosing a large enough constant $\alpha$, since $\eigvec^T F \eigvec < 0$. Finally, plugging this control input into Equation \eqref{eq:feas_lemma1_2} yields $|\socc \eigvec|\  \alpha + \socd \geq 0$, which again holds for a sufficiently large $\alpha$. Therefore, \eqref{eq:gp-cbf-clf-socp} is feasible by Lemma \ref{lemma:fundamental}.
\end{proof}

The crux of this sufficient condition is that a single scalar value ($\lambda_\dagger$) being negative guarantees the feasibility of the problem, and this condition can be easily checked online before solving \eqref{eq:gp-cbf-clf-socp}. 
\fc{Note that the value of $\eigval$ for each state is dependent on the choice of training data. Therefore, we foresee many useful applications of the condition $\eigval<0$ to verify the data in use. For instance, this could serve to implement online adaptive data-selection algorithms. Moreover, the proof of Lemma \ref{lemma:sufficient} is constructive. In fact, \eqref{eq:feasibledirection} constitutes the control input direction that best balances safety and uncertainty. }

\subsection{Necessary and Sufficient Condition}

Finally, we state the necessary and sufficient condition for pointwise feasibility. This condition includes those of Lemmas \ref{lemma:necessary} and \ref{lemma:sufficient} with some additional cases.

\begin{figure}%[!htb]
\centering
\includegraphics[width=0.48\textwidth]{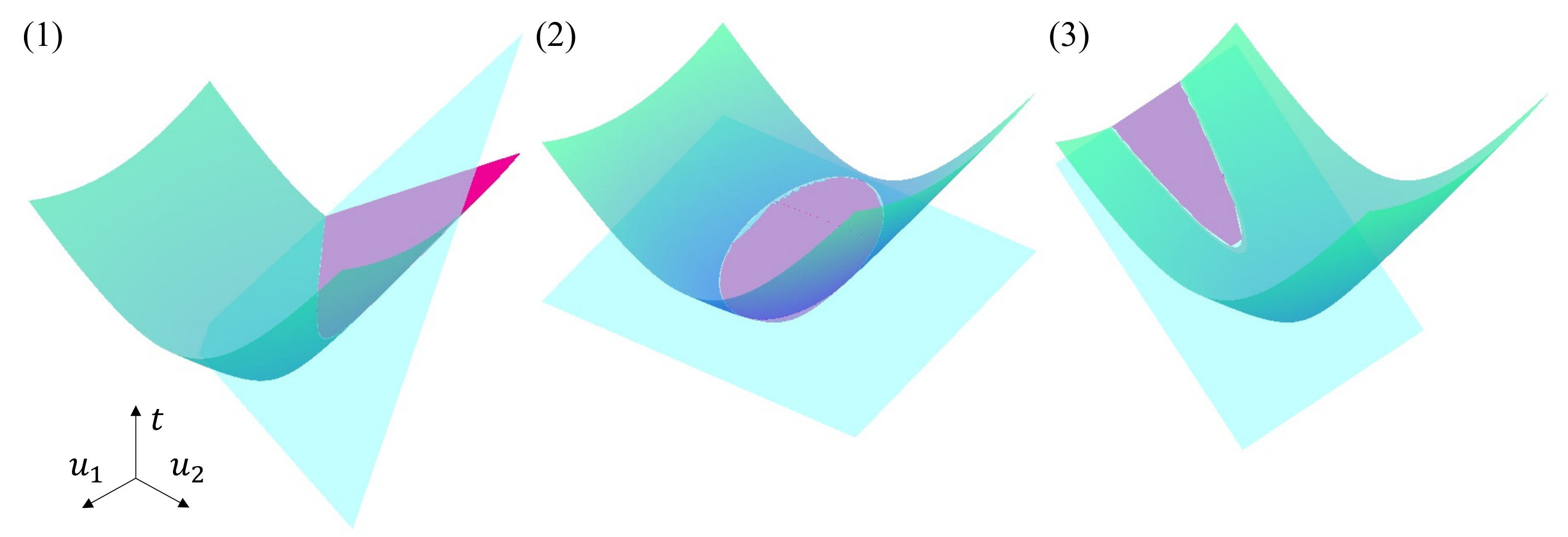}
\vspace{-7mm}
\caption{Visualization of the feasibility conditions of Theorem \ref{th:nec_and_suf}. The green surface is the hyperboloid $\norm{\socA u + \socb}\! =\! t$, the blue hyperplane is $\socc u\! +\! \socd\! =\! t$, and the pink region indicates the feasible set. Case 1) hyperbolic intersection, Case 2) elliptical intersection, and Case 3) parabolic intersection. Note that for Cases 2) and 3), if \eqref{eq:case2-add} and \eqref{eq:case3-add} are not satisfied, respectively, the feasible set is empty.}
\label{fig:feasibility}
\vspace{-10pt}
\end{figure}

\begin{theorem}[Necessary and sufficient condition for pointwise feasibility]
\label{th:nec_and_suf}
For a given state $x \in \R^n$, let $\eigval$ be the minimum eigenvalue of the symmetric matrix $F$ defined in \eqref{eq:defineF}. The GP-CBF-CLF-SOCP \eqref{eq:gp-cbf-clf-socp} is feasible at $x$ if and only if \eqref{eq:necessary} is satisfied and one of the following cases holds:
\begin{enumerate}
    \item $\eigval <0$;
    \item $\eigval>0$,
    and \vspace{-.5em}
    \begin{equation}
    \label{eq:case2-add}
        \socd-\socc F^{-1}\{\socA^T\socb-\socc^T\socd\}\!\geq\!0; \vspace{-.5em}
    \end{equation}
    \item $\eigval =0$,
    and \vspace{-.5em}
    \begin{equation}
    \label{eq:case3-add}
        \socd - \socc\left(\socA^T\socA \right)^{-1}\socA^T \socb^T > 0. \vspace{-.5em}
    \end{equation}
\end{enumerate}
\end{theorem}
 Case 1) is exactly Lemma \ref{lemma:sufficient}, and it corresponds to the feasible set being hyperbolic. Cases 2) and 3) correspond to elliptic and parabolic feasible sets, respectively (Fig. \ref{fig:feasibility}).
\begin{proof} See Appendix. \end{proof}

We believe that Theorem \ref{th:nec_and_suf} constitutes the first step towards understanding what conditions the distribution of data should satisfy in order to obtain probabilistic safety guarantees in the presence of actuation uncertainty. This problem is related to the notion of persistency of excitation \cite{verhaegen2007filtering}. However, we think that Theorem \ref{th:nec_and_suf} hints at less restrictive conditions that are sufficient to maintain safety. Verifying specific conditions on the data distribution remains as our future work. Finally, note that the presented analysis is also directly applicable to the GP-CLF-SOCP of \eqref{eq:gp-clf-socp} by considering the CLF chance constraint instead of the CBF one.
\section{Numerical Results}
\label{sec:06results}
To validate our proposed controller we utilize the following model of an adaptive cruise control system: \vspace{-1em}

\small

\begin{equation}
    \dot{x} = f(x) + g(x)u, ~~f(x)\! =\! \begin{bmatrix} - F_r(v)/m \\ v_0 - v \end{bmatrix},\
    g(x)\! =\! \begin{bmatrix} 0 \\ 1/m \end{bmatrix},
\end{equation}

\normalsize
\noindent where $x\!=\![v\ z]^T\in \R^2$ is the system state, with $v$ being the forward velocity of the ego car, and $z$ the distance between the ego car and the front car. $u \in \R$ is the ego car's wheel force as control input, $v_0$ is the velocity of the front car which is assumed to be constant (14 m/s), $m$ is the mass of the ego car, and $F_r(v)\!=\!f_0\!+\!f_1 v\!+\!f_2 v^2$ is the rolling resistance acting on the ego car.

The control objective is to reach a desired speed command of $v_d = 24$ m/s while maintaining a safe distance $z \geq T_h v$ with respect to the front vehicle, where $T_h\!=\!1.8$ s is the lookahead time. For this, we design a CLF as $V(x) = (v-v_d)^2$ and a CBF as $B(x)\!=\!z - T_h v$.

In order to introduce model uncertainty,
for the nominal model we use $m\!=\!1650\ \text{kg},\ f_0\!=\!0.1, \ f_1\!=\!5, \ f_2\!=\!0.25;$ and for the true plant $m\!=\!3300\ \text{kg},\ f_0\!=\!0.2, \ f_1\!=\!10, \ f_2\!=\!0.5$.

\begin{figure}
\centering
\includegraphics[width=\columnwidth]{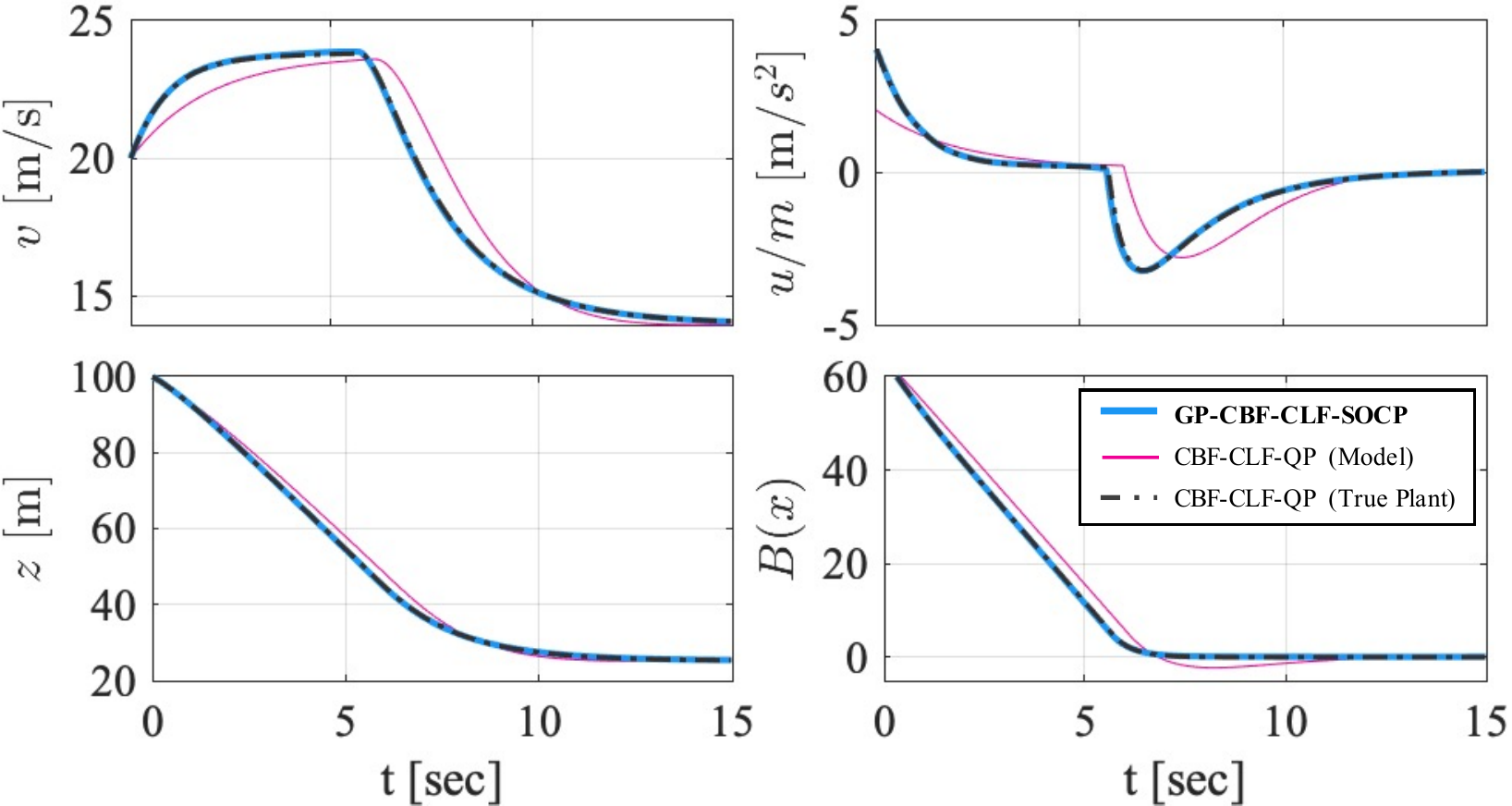}
\vspace{-7mm}
\caption{Simulation results of applying the GP-CBF-CLF-SOCP (blue) to the adaptive cruise control example under model uncertainty, compared to the nominal-model-based CBF-CLF-QP (pink) and the oracle true-plant-based CBF-CLF-QP (grey dashed). Note that the simulation ends with $v$ decreased since it is approaching the front vehicle, whose speed (14 m/s) is lower than the desired speed (24 m/s).
\fc{The proposed GP-CBF-CLF-SOCP closely matches the performance of the oracle CBF-CLF-QP and always respects safety $(B(x)\! \geq\! 0)$, whereas the nominal-model-based CBF-CLF-QP violates the safety constraint.}
}
\label{fig:results}
\vspace{-.8em}
\end{figure}

\begin{figure}%[!htb]
\centering
\includegraphics[width=\columnwidth]{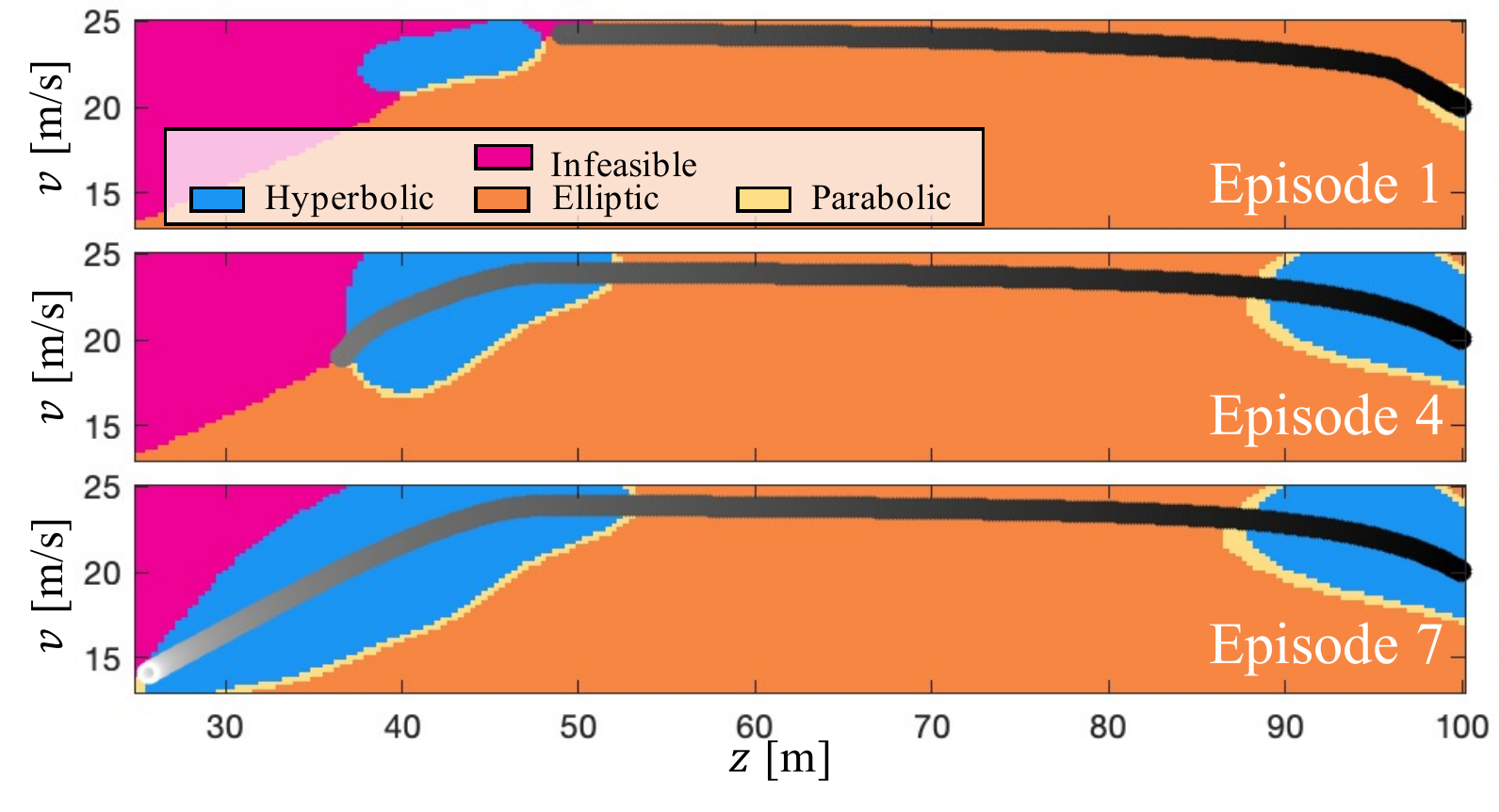}
\vspace{-9mm}
\caption{Evolution of the feasible region during the learning episodes for the adaptive cruise control example. The trajectories evolve from black to white. As more data is collected, the feasible set gets larger. We ran a total of 7 episodes, and the final number of data points is $N=189$.}
\label{fig:episodes}
\vspace{-1em}
\end{figure}

As it can be seen in the lower right plot of Fig. \ref{fig:results}, a CBF-CLF-QP controller constructed based on the nominal model (depicted in pink) violates the safety constraint after about 7 seconds of simulation (since $B(x) < 0$).
This is caused by the model uncertainty.
In order to correct this, we use the GP regression approach that has been presented in this paper for both CLF and CBF constraints.

The data collection process is conducted in an episodic fashion. We start by running an initial rollout of the nominal CBF-CLF-QP until the system violates the safety constraint. With that data, we obtain a first GP model that is now used to run a rollout with the GP-CBF-CLF-SOCP. We stop the rollout if the optimization problem becomes infeasible or if the system becomes unsafe. Because of the high probability bound ($1-\delta$ = 0.95) we use for \eqref{eq:V_UCB} and \eqref{eq:B_UCB}, infeasibility always occurs before the system becomes unsafe.
We then add the collected data to the GP model and run the next rollout. This process is repeated until the rollout is completely feasible. Fig. \ref{fig:episodes} represents this process, and it also depicts how the feasibility region evolves as more data is collected.

In Fig. \ref{fig:results} the performance of the GP-CBF-CLF-SOCP controller of the final episode is shown in blue. For reference, the results of the simulation with the oracle CBF-CLF-QP controller that uses the true plant dynamics are also illustrated in the plots. The proposed GP-CBF-CLF-SOCP recovers the performance of the oracle CBF-CLF-QP and always respects safety, whereas the nominal model-based CBF-CLF-QP violates the safety constraint due to model uncertainty. Finally, the computation time of the SOCP per timestep was $7.4\pm4.2$ ms on a laptop with an
Intel Core i7 and 32 GB of RAM, fast enough for real-time applications.
\section{Conclusion}
\label{sec:07conclusion}
We presented a framework to construct a GP-based safe stabilizing controller (GP-CBF-CLF-SOCP) that takes into account the problem of model uncertainty in the CLF and CBF constraints.
We also presented an analysis of pointwise feasibility of the proposed
GP-CBF-CLF-SOCP, which we believe should pave the way to obtain 
safety guarantees for systems with uncertain dynamics and input effects without requiring the infeasible collection of data that fully characterize the dynamics of the system. This will be the focus of our future work.

\vspace{-1mm}
\section*{Acknowledgements}
We would like to thank Andrew Taylor, Victor Dorobantu, Ivan Jimenez Rodriguez, and Hotae Lee for the insightful discussions.

\appendix
\noindent \textbf{Proof of Lemma \ref{lemma:necessary}}\\
Positive definiteness of $H(x)$ indicates that there does not exist any $u \in \R^m$ such that $[1\ u^T] H(x) \begin{bmatrix} 1\\ u \end{bmatrix} \leq 0$,
and this is a contradiction to Lemma \ref{lemma:fundamental}. Therefore, $H(x)$ cannot be positive definite if the GP-CBF-CLF-SCOP is feasible.

The rest of the proof shows that condition \eqref{eq:necessary} is equivalent to $H(x)$ not being positive definite. Let $\psi(x) \coloneqq [\socd\;\socc]$. Condition \eqref{eq:necessary} does not hold if and only if
\vspace{-.5em}
\begin{equation}
\label{eq:necessary_converse}
\vspace{-.5em}
1- \psi(x) \frac{1}{\beta^2}\gpGramB(x)^{-1}\psi(x)^T = M/(\beta^2 \gpGramB(x)) > 0,
\end{equation}
where $M = \begin{bmatrix} 1 & \psi(x)\\ \psi(x)^T &  \beta^2 \gpGramB(x)\end{bmatrix}$. We use the operator $/$ for the Schur complement.
From \cite[Thm. 1.12]{zhang2006schur}, since $\gpGramB(x)$ is positive definite, \eqref{eq:necessary_converse} holds if and only if $M$ is positive definite.
Applying again the same theorem, but this time to $M/1$, \eqref{eq:necessary_converse} is equivalent to $M/1 = \beta^2 \gpGramB(x) - \psi(x)^T\psi(x)$\;$=[\socb \ \socA]^T [\socb \ \socA]-\psi(x)^T\psi(x)\!=\!H(x)$ being positive definite. Therefore, \eqref{eq:necessary} holds if and only if $H(x)$ is not positive definite.
\\

\noindent \textbf{Proof of Theorem \ref{th:nec_and_suf}}\\
We start the proof by providing a geometric interpretation of the SOC constraint. For a fixed $x \in \R^n$,  $\norm{\socA u + \socb}\! =\! t$ is the positive sheet of an $m$-dimensional hyperboloid in $\R^{m+1}$ (illustrated in Fig. \ref{fig:feasibility}).
This surface 
asymptotically converges to the conical surface $\norm{\socA(u-u_0)}\! =\! t$ as $\norm{u} \xrightarrow{} \infty$, where \vspace{-.5em}
\begin{equation}
\label{eq:least-square-u}
    u_0 = - \left(\socA^T \socA \right)^{-1} \socA^T \socb \vspace{-.3em}
\end{equation}
is the $u$ that minimizes $\norm{\socA u + \socb}$, given by the least-squares solution. We will refer to the conical surface $\norm{\socA(u-u_0)}\! =\! t$ as the \emph{asymptote} of
$\norm{\socA u + \socb}\! =\! t$.

Since $Q(x) \succ 0$, the problem \eqref{eq:gp-cbf-clf-socp} is feasible if and only if an intersection between the hyperboloid $\norm{\socA u\! +\! \socb}\! =\! t$ and the hyperplane $\socc u\! +\! \socd\! =\! t$ exists.

Case 1) Lemma \ref{lemma:sufficient} is the proof. Note that this condition itself implies that \eqref{eq:necessary} is satisfied. The slope of the hyperplane $\socc u + \socd = t$ is greater than the slope of the asymptote of the hyperboloid for some direction of $u$ (Fig. \ref{fig:feasibility}.(1)).

Case 2) Since the smallest eigenvalue of $F$ is positive, this means that $F\succ 0$. Note that the lower right block of the matrix $H(x)$ in \eqref{eq:H_def} is $F$. Therefore, $F\succ 0$ implies that the left-hand side of Equation \eqref{eq:feas_lemma1_1} is strictly convex, attaining the global minimum at some $u = u_1 \in \R^m$. The first order optimality condition gives \vspace{-.5em}
\begin{equation*}
     u_1 = -F^{-1}h, \vspace{-.5em}
\end{equation*}
with $h := \socA^T\socb - \socc^T\socd$. Since at $u_1$ the minimum is achieved, Equation \eqref{eq:feas_lemma1_1} is satisfied if and only if \vspace{-.3em}
\begin{equation}
\label{eq:condition_u1}
    [1\ u_1^T] H(x) \begin{bmatrix} 1 \\u_1 \end{bmatrix} \leq 0. \vspace{-.3em}
\end{equation}
Plugging \eqref{eq:H_def} and $u_1 = -F^{-1}h$ into \eqref{eq:condition_u1}, we get \vspace{-.3em}
\begin{equation}
\label{eq:condition_h}
    \left(\socb^T\socb - \socd^2 \right) - h^T F^{-1} h = H(x) / F \leq 0.\vspace{-.5em}
\end{equation}

\noindent Since $H(x)$ cannot be positive definite by the necessary condition \eqref{eq:necessary}, and $F$ is positive definite, from \cite[Thm. 1.12]{zhang2006schur} \eqref{eq:condition_h} must be satisfied.
Consequently, \eqref{eq:feas_lemma1_1} holds for $u=u_1$.
Now, plugging $u_1$ into $\eqref{eq:feas_lemma1_2}$ we have: $\socc u_1 + \socd = \socd - \socc F^{-1}\{\socA^T\socb - \socc^T\socd\}$, which is precisely the left-hand side of \eqref{eq:case2-add}. The feasible region is non-empty if and only if the above expression is greater than or equal to 0 due to Lemma~\ref{lemma:fundamental} (Fig. \ref{fig:feasibility}.(2)). If the above expression is smaller than 0, it means that the hyperplane $\socc u \! + \! \socd \! =\! t$ intersects with the hyperboloid's negative sheet, $-\norm{\socA u \! + \! \socb}\!=\!t$, forming an ellipse, and cannot intersect with the positive sheet.
Therefore, when $\eigval > 0$, the SOCP \eqref{eq:gp-cbf-clf-socp} is feasible if and only if \eqref{eq:case2-add} is satisfied.

Case 3) Note that $\eigval\!=\!0$ means that the hyperplane $\socc u\! +\! \socd\! =\! t$ and the asymptote of $\norm{\socA u + \socb}\! =\! t$ have the same slope for some direction of $u$ (Fig. \ref{fig:feasibility}.(3)).
Define \vspace{-.3em}
\begin{equation}
    p \coloneqq \socd - \socc\left(\socA^T\socA\right)^{-1}\socA^T \socb^T. \vspace{-.3em}
\end{equation}
Then, the given condition \eqref{eq:case3-add} holds if and only if $\!p\!>\!0$. Consider $u=u_0$ from \eqref{eq:least-square-u} that minimizes $\norm{\socA u + \socb}$. Then, $p = \socd + \socc u_0 $. Let $\eigvec$ be the unit eigenvector of $F$ associated with the eigenvalue $\eigval=0$. Then, $\eigvec^T F \eigvec = 0$. Since $\socA^T \socA$ is positive definite, this means that $\socc \eigvec \neq 0$. Now let's consider a control input of the form \vspace{-.5em}
\begin{equation}
\label{eq:soc-case3-u}
    u = u_0 + \alpha \cdot \text{sgn}(\socc \eigvec) \cdot \eigvec, \quad \alpha > 0. \vspace{-.5em}
\end{equation}
Using this control law, the left-hand side of \eqref{eq:feas_lemma1_1} becomes \vspace{-1em}

\small
\begin{align}
\label{eq:soc-proof-sub1}
    & \{\socb^T \socb\!-\!\socd^2\!+\!2 h^T u_0\!+\!u_{0}^T F u_0\} - 2 \alpha \!\cdot\! p \!\cdot\! |\socc \eigvec|,
\end{align}
\normalsize

\noindent and the left-hand side of \eqref{eq:feas_lemma1_2} becomes \vspace{-.3em}
\begin{align}
\label{eq:soc-proof-sub2}
    % \socc(u_0 + k \eigvec) + \socd = 
    p + \alpha \cdot |\socc \eigvec|. 
    \vspace{-.8em}
\end{align}
When $p > 0$, there exists a sufficiently large $\alpha$ such that \eqref{eq:soc-proof-sub1} becomes non-positive and \eqref{eq:soc-proof-sub2} becomes positive. Therefore, from Lemma~\ref{lemma:fundamental}, the SOCP \eqref{eq:gp-cbf-clf-socp} is feasible. 

Note that the geometric interpretation of the condition $p\!>\!0$ is that the hyperplane $\socc u\! +\! \socd\! =\! t$,
which has the same slope as the asymptote of $\norm{\socA u\! +\! \socb}\! =\! t$
along $\eigvec$, should be placed over the asymptote in order for it to intersect the positive sheet of the hyperboloid.
At $u=u_0$, the asymptote $\norm{\socA(u-u_0)}\! =\! t$ takes value $t\!=\!0$. $p$ is the value of the hyperplane at $u\! =\! u_0$. Therefore, when $p \le 0$, the hyperplane is always under the positive sheet of the hyperboloid, and never intersects it. Consequently, the constraint \eqref{eq:constraint_std} is not feasible.

\addtolength{\textheight}{-12cm}   % This command serves to balance the column lengths

%%%%%%%%%%%%%%%%%%%%%%%%%%%%%%%%%%%%%%%%%%%%%%%%%%%%%%%%%%%%%%%%%%%%%%%%%%%%%%%%

%%%%%%%%%%%%%%%%%%%%%%%%%%%%%%%%%%%%%%%%%%%%%%%%%%%%%%%%%%%%%%%%%%%%%%%%%%%%%%%%

%%%%%%%%%%%%%%%%%%%%%%%%%%%%%%%%%%%%%%%%%%%%%%%%%%%%%%%%%%%%%%%%%%%%%%%%%%%%%%%%

% % \balance
\bibliographystyle{IEEEtran}
\bibliography{reference.bib}{}
\end{document}